\documentclass[journal,twocolumn]{IEEEtranTCOM}

\usepackage{array,cite}
\usepackage{graphicx,epsfig,subfigure}
\usepackage{amsthm}
\usepackage{amsmath}
\usepackage{mdwmath}
\usepackage{array}
\usepackage{subeqnarray}
\usepackage{stfloats}
\usepackage{algorithm,algorithmic}
\usepackage{xcolor}
\usepackage{amssymb}
\usepackage{citesort}

\newtheorem{proposition}{Proposition}
\newtheorem{lemma}{Lemma}
\newtheorem{conjecture}{Conjecture}

\begin{document}

\title{A Split-Reduced Successive Cancellation List Decoder for Polar Codes}
\author{\IEEEauthorblockN{Zhaoyang~Zhang, \textit{Member, IEEE}, Liang~Zhang, Xianbin~Wang, \\
Caijun~Zhong, \textit{Senior Member, IEEE}, and H. Vincent~Poor, \textit{Fellow, IEEE}}
\thanks{This work was supported in part by the National Key Basic Research Program of China under Grant 2012CB316104, the National Hi-Tech R\&D Program of China under Grant 2014AA01A702, the National Natural Science Foundation of China under Grant 61371094, the Zhejiang Provincial Natural Science Foundation under Grant LR12F01002, the Huawei HIRP Flagship Project under Grants YB2015040053 and YB2013120029, and the U. S. National Science Foundation under Grants CCF-1420575 and ECCS-1343210.}\\
\thanks{Z. Zhang (email: {\tt ning\_ming@zju.edu.cn}), L. Zhang (email: {\tt 0705zhangliang@sina.com}), X. Wang (email: {\tt wangxianbin@outlook.com}), and C. Zhong (email: {\tt caijunzhong@zju.edu.cn}) are with the College of Information Science and Electronic Engineering, Zhejiang University, China. H. V. Poor (email: {\tt poor@princeton.edu}) is with the School of Engineering and Applied Science, Princeton University.}
}

\maketitle

\begin{abstract}
This paper focuses on low complexity successive cancellation list (SCL) decoding of polar codes. In particular, using the fact that splitting may be unnecessary when the reliability of decoding the unfrozen bit is sufficiently high, a novel splitting rule is proposed. Based on this rule, it is conjectured that, if the correct path survives at some stage, it tends to survive till termination without splitting with high probability. On the other hand, the incorrect paths are more likely to split at the following stages. Motivated by these observations, a simple counter that counts the successive number of stages without splitting is introduced for each decoding path to facilitate the identification of correct and incorrect path. Specifically, any path with counter value larger than a predefined threshold $\omega$ is deemed to be the correct path, which will survive at the decoding stage, while other paths with counter value smaller than the threshold will be pruned, thereby reducing the decoding complexity. Furthermore, it is proved that there exists a unique unfrozen bit $u_{N-K_1+1}$, after which the successive cancellation decoder achieves the same error performance as the maximum likelihood decoder if all the prior unfrozen bits are correctly decoded, which enables further complexity reduction. Simulation results demonstrate that the proposed low complexity SCL decoder attains performance similar to that of the conventional SCL decoder, while achieving substantial complexity reduction.
\end{abstract}


\section{Introduction}
Polar codes, first discovered by Ar\*{\i}kan \cite{firstArikan}, are the first capacity-achieving codes for binary-input discrete memoryless channels with an explicit and deterministic structure. In addition, it was shown that a simple successive cancellation (SC) decoder asymptotically achieves the capacity with low complexity, of order $O(N\mathrm{log}N)$ where $N$ is the block-length \cite{firstArikan}. Due to these extraordinary properties, polar codes have captured the attention of both academia and industry alike.

Motivated by the fact that the SC decoder tends to exhibit less promising performance with finite-length block codes, an important line of current research is to seek efficient decoders with better performance for polar codes. In \cite{list_decoding_Vardy} and \cite{SCL_decoding_Chenkai}, the authors proposed the successive cancellation list (SCL) decoder, which was shown to approach the performance of maximum-likelihood (ML) decoding in the high signal-to-noise ratio (SNR) regime, albeit at the cost of higher processing complexity of $O(LN\mathrm{log}N)$, where $L$ is the list size. Later in \cite{Adaptive_SCL_Libin}, it was further demonstrated that polar codes concatenated with a high rate cyclic redundancy check (CRC) code outperform turbo and LDPC codes by applying an adaptive SCL decoder with sufficiently large list size. Trading storage complexity for computational reduction, the authors in \cite{Stack_decoding_Chenkai} and \cite{Sequential_decoding_Trifonov} proposed the successive cancellation stack (SCS) decoder, which was shown to have much lower computational complexity compared with the SCL decoder, especially in the high SNR regime, where its complexity becomes close to that of the SC decoder. More recently, a novel successive cancellation hybrid decoder was proposed in \cite{Improved_SCL_Chenkai}, which essentially combines the ideas of SCL and SCS decoders and provides a fine balance between the computational complexity and storage complexity.

As discussed above, the SCL decoder achieves superior performance compared to the SC decoder at the price of increased complexity, especially when the list size $L$ is very large, which has prohibited its widespread implementation in practice. As such, reducing the computational complexity of the SCL decoder is of considerable importance, motivating the current research. For the conventional SCL decoder, each decoding path will be split into two paths when decoding an unfrozen bit and the number of ``best paths'' remains at $L$ until the termination of decoding, which causes an increased complexity of $O(LN\mathrm{log}N)$. To reduce the decoding complexity, we argue that it is unnecessary to split all the decoding paths, supported by the key observation that splitting can be avoided if the reliability of deciding the unfrozen bit $u_i=0$ or $u_i=1$ is sufficiently high. A direct consequence of such a split-reduced approach is that many fewer paths are likely to survive after pruning, i.e., the number of ``best paths'' is much smaller than the list size, which results in further complexity reduction.

The main contributions of this paper are summarized as follows:
\begin{enumerate}
  \item Taking advantage of the fact that splitting is unnecessary if the unfrozen bit can be decoded with high reliability, a novel splitting rule is defined. Moreover, the behavior of the correct and incorrect decoding paths are characterized under the new splitting rule. Based on which, a split-reduced SCL decoder is proposed. By avoiding unnecessary path splitting as well as efficiently reducing the number of surviving paths, the proposed split-reduced SCL decoder can achieve significant reduction of complexity while retaining a similar error performance compared with the conventional SCL decoder.
  \item Furthermore, we prove the existence of a particular unfrozen bit $u_{N-K_1+1}$, after which the SC decoder achieves the same error performance as the ML decoder if all the prior unfrozen bits are correct, and show how to locate the particular unfrozen bit. Then, exploiting this crucial property, an enhanced version of the split-reduced SCL decoder is proposed.
\end{enumerate}

The rest of the paper is organized below. In Section \uppercase\expandafter{\romannumeral2}, we provide some basic concepts and notation for polar codes and the SCL decoder. In Section \uppercase\expandafter{\romannumeral3}, we present a novel split-reduced SCL decoder and provide an analysis of its decoding behavior. An enhanced version of the split-reduced SCL decoder is proposed in Section \uppercase\expandafter{\romannumeral4} while the simulation results are provided in Section \uppercase\expandafter{\romannumeral5}. Finally, Section \uppercase\expandafter{\romannumeral6} gives a brief summary of the paper.

\section{Preliminaries and Notations}
In this section, we provide a brief introduction to polar codes, the SC decoder and the SCL decoder, and explain the notation adopted in the paper.
\subsection{Polar Codes}
For a polar code with block-length $N=2^n$ and dimension $K$, the generator matrix can be written as $G_N=B_NG_2^{\otimes n}$, where $G_2=\left[
\begin{array}{ccc}
1 & 0 \\
1 & 1
\end{array}
\right]$,
$B_N$ is an $N\times N$ bit-reversal permutation matrix, and $(\cdot)^{\otimes n}$ denotes the $n$-th Kronecker power. We use $a_i^j$ to represent the sequence $(a_i,a_{i+1},...,a_j)$, and as such, any codeword of a polar code can be expressed as $c_1^N=u_1^NG_N$, where $u_1^N$ is the information sequence consisting of $N-K$ frozen bits and $K$ unfrozen bits.

Let $W: \mathcal{X}\rightarrow \mathcal{Y}$ denote a binary discrete memoryless channel with input alphabet $\mathcal{X} = \{0, 1\}$, output alphabet $\mathcal{Y}$, and channel transition probabilities $\{W(y|x):x\in \mathcal{X}, y\in \mathcal{Y}\}$. After channel polarization, the transition probability of the $i$-th subchannel is given by
\begin{equation}\nonumber
\begin{aligned}
W_N^{(i)}(y_1^N,u_1^{i-1}|u_i)=\sum_{u_{i+1}^N \in \mathcal{X}^{N-i}}\frac{1}{2^{N-1}}W_N(y_1^N|u_1^N),
\end{aligned}
\end{equation}
where
\begin{equation}\nonumber
\begin{aligned}
W_N(y_1^N|u_1^N)=\prod_{i=1}^NW(y_i|x_i).
\end{aligned}
\end{equation}

To implement the encoding, the $K$ most reliable subchannels are selected to transmit the unfrozen bits while the remaining subchannels are used for sending the frozen bits which are set to some fixed values (see \cite{how_to_construct_Vardy}). Without loss of generality, we assume that the frozen bits are zero valued.

\subsection{SC and SCL Decoding}
Define the logarithmic likelihood ratio (LLR) of $u_i$ as
\begin{equation}\nonumber
\begin{aligned}
L(u_i)=\mathrm{log}\frac{W_N^{(i)}(y_1^N,\hat{u}_1^{i-1}|u_i=0)}{W_N^{(i)}(y_1^N,\hat{u}_1^{i-1}|u_i=1)},
\end{aligned}
\end{equation}
where $\hat{u}_j$ and $y_1^N$ denote an estimate of $u_j$ and the received sequence from the channel, respectively. We use base-$e$ logarithms throughout this paper unless otherwise specified.

For standard SC decoding, bit-by-bit information decoding is performed. As such, if $u_i$ is an unfrozen bit, $\hat{u}_i$ is set to either $0$ or $1$ according to the sign of $L(u_i)$, i.e.,
\begin{eqnarray}
& \hat{u}_i=\left\{
 \begin{aligned}
    &0, \ \text{if} \ L(u_i) > 0,\\
    &1, \ \text{if} \ L(u_i) < 0.\\
 \end{aligned}
 \right.
 \label{LLR_for_SC}
\end{eqnarray}

Unlike the SC decoder which employs a hard-decision for each bit, the SCL decoder inspects both options for the estimate of any unfrozen bit $u_i$ and splits each decoding path into two paths. Nevertheless, at each decoding stage, only the best $L$ paths survive in order to reduce the complexity.

\section{A simple split-reduced SCL decoder}
This section presents a simple split-reduced SCL decoder. We start by first introducing a new splitting rule, and then examine the error performance under this rule. Based on this, a novel SCL decoding algorithm is proposed. Finally, a brief discussion of the complexity comparison between the proposed algorithm and the conventional SCL decoding algorithm is provided.

\subsection{The Splitting Rule}
As mentioned above, the proposed split-reduced SCL decoder exploits the fact that splitting is unnecessary if the reliability of decoding the unfrozen bit is high enough. Therefore, to implement such a decoder, the first step is to define the rule of splitting, i.e., how to decide whether the current decoding path shall split or not, and under what conditions. In the following, we first choose a metric to measure the decoding reliability, then define an appropriate threshold for this metric to establish the rule.

According to polarization, each unfrozen bit $u_i$ would observe a subchannel $W_N^{(i)}(y_1^N,u_1^{i-1}|u_i)$ and can be considered that $u_i$ is transmitted through such a synthetic channel. Thus, the reliability of decoding $u_i$ actually depends on $W_N^{(i)}$. Although for binary erasure channel (BEC), the reliability can be explicitly described by $Z(W)$ (see \cite{firstArikan}) and computed in a recursive manner, the same approach does not appear to be applicable to other channels including binary symmetric channel (BSC). Therefore, we adopt the \emph{a posteriori} probability as the metric of reliability for each subchannel, mainly inspired by \cite{Efficient_design_Trifonov}, where the Gaussian approximation was used to give an estimate for the error probability of $W_N^{(i)}$.

Having determined the measure of reliability, we now proceed to find an appropriate threshold. For subchannel $W_N^{(i)}$ and any given input $u_i$, suppose that all prior bits have been correctly decoded. Now let $P_e(u_i)$ denote the estimation error probability of $u_i$ averaged over all possible outputs $(y_1^N, u_1^{i-1})$, i.e.,
\begin{equation}\nonumber
\begin{aligned}
&P_e(u_i)=P(\hat{u}_i=u_i\oplus 1)\\
&=\sum_{u_1^{i-1} \in \mathcal{X}}\sum_{y_1^N \in \mathcal{Y}}P((1-2u_i)L(u_i)<0|\hat{u}_1^{i-1}=u_1^{i-1},u_i,y_1^N),
\end{aligned}
\end{equation}
then $P_e(u_i)$ in fact describes the probability that $u_i$ is incorrectly estimated in terms of the subchannel $W_N^{(i)}(y_1^N,u_1^{i-1}|u_i)$, given the correct prior bits $u_1^{i-1}$. Once each $P_e(u_i)$ is computed, one has obtained some `prior' knowledge which implies that if the correct path reaches stage $i$ (decode $u_i$), the probability that $u_i$ is correctly estimated should not be too smaller than $1-P_e(u_i)$. In other words, $1-P_e(u_i)$ can be regarded as the confidence level of decoding reliability of the $i$-th subchannel. Hence, it is a natural choice for threshold. In general, analytical evaluation of $P_e(u_i)$ is difficult. Nevertheless, it can be computed via Monte Carlo simulation or the method introduced in \cite{how_to_construct_Vardy}.

For the particular case of additive white Gaussian noise (AWGN) channels, $P_e(u_i)$ can also be evaluated by assuming that the LLR follows Gaussian distribution with mean $\mu$ and variance $\sigma^2 = 2|\mu|$ \cite{Efficient_design_Trifonov,polar_codes_over_AWGN_channel,Design_of_LDPC_Shokrollahi,Analysis_of_sum-product_Richardson}. While the Gaussian approximation assumption is used for analytical tractability, there are in fact theoretical supports to corroborate such assumption, as elaborated in the following.

Without loss of generality, assuming an all-zero codeword is transmitted over an AWGN channel with noise variance $\sigma_n^2$ using binary phase shift keying (BPSK), i.e., the codeword $c_1^N$ is mapped to signal $x_1^N$ by $x_i=1-2c_i$, it is easy to show that the LLR $L(y_i)$ of each received symbol $y_i$ follows the $\mathcal{N}(\frac{2}{\sigma_n^2},\frac{4}{\sigma_n^2})$ distribution.

Recall that equations (75) and (76) in \cite{firstArikan} can be rewritten from an LLR perspective as
\begin{equation}\nonumber
\begin{aligned}
&L_N^{(2i-1)}(y_1^N,u_1^{2i-2})\\
&=L_{N/2}^{(i)}(y_1^{N/2},u_{1,o}^{2i-2}\oplus u_{1,e}^{2i-2}) \boxplus
L_{N/2}^{(i)}(y_{N/2+1}^{N},u_{1,e}^{2i-2})
\end{aligned}
\end{equation}
and
\begin{equation}\nonumber
\begin{aligned}
&L_N^{(2i)}(y_1^N,u_1^{2i-1})\\&=L_{N/2}^{(i)}(y_1^{N/2},u_{1,o}^{2i-2}\oplus u_{1,e}^{2i-2})+
L_{N/2}^{(i)}(y_{N/2+1}^{N},u_{1,e}^{2i-2}),
\end{aligned}
\end{equation}
where $a \boxplus b=\mathrm{log}\frac{1+e^{a+b}}{e^a+e^b}$. Note that we have used $u_i$ instead of the estimate $\hat{u}_i$ since the real values of $u_1^{i-1}$ are provided when we compute $P_e(u_i)$ and the coefficient $(1-2u_{2i-1})$ in front of $L_{N/2}^{(i)}(y_1^{N/2},u_{1,o}^{2i-2}\oplus u_{1,e}^{2i-2})$ is omitted as well since all-zero codeword is transmitted. To simplify notations, we denote $f_1(a,b)=a \boxplus b$ and $f_2(a,b)=a+b$. It was demonstrated in \cite{Analysis_of_sum-product_Richardson} that if the symmetry condition, which can be expressed as $f(x)=f(-x)e^x$ with $f(x)$ being the density of an LLR message, is satisfied, the probability density function (pdf) of the output of the check node is approximately a Gaussian density which satisfies the symmetry condition as well. Therefore, if both $a$ and $b$ are Gaussian random variables that satisfy the symmetry condition, according to the result of \cite{Analysis_of_sum-product_Richardson}, $f_1(a,b)$ is approximately Gaussian distributed and satisfies the symmetry condition. Also, if $a$ and $b$ have the same pdf, i.e., $\mathcal{N}(m,2m)$, it is easy to check that $f_2(a,b)$ follows Gaussian distribution with $\mathcal{N}(2m,4m)$ and satisfies the symmetry condition as well.

Now let us take a look at the received LLR $L(y_i) \sim \mathcal{N}(m,2m)$ where $m=\frac{2}{\sigma_n^2}$, and it is easy to show that
\begin{equation}\nonumber
\begin{aligned}
\frac{1}{\sqrt{4\pi m}}e^{-\frac{(-y_i-m)^2}{4m}}e^{y_i}
&=\frac{1}{\sqrt{4\pi m}}e^{-\frac{(-y_i-m)^2-4m}{4m}}\\
&=\frac{1}{\sqrt{4\pi m}}e^{-\frac{(y_i-m)^2}{4m}},
\end{aligned}
\end{equation}
which indicates that the density of all the received LLR messages satisfy the symmetry condition. With some simple algebraic manipulations, it can be shown that each LLR $L(u_i)$ can be expressed as a compound function of $f_1$ and $f_2$ with $\{L(y_1),L(y_2),...,L(y_N)\}$ as the input. Since $L(y_i)$ has the same pdf $\mathcal{N}(m,2m)$, it is easy to verify that all the intermediate outcomes of $f_1(a,b)$ and $f_2(a,b)$ are approximately Gaussian distributed and satisfy the symmetry condition. Therefore, $L(u_i)$ can be approximated by the Gaussian distribution.

Now by expressing equations (75) and (76) in [1] in the form of expectation, we have \cite{Analysis_of_sum-product_Richardson}:
\begin{eqnarray}
\begin{aligned}\centering\label{ELu}
&\textbf{E}[L(u_1)]\\&=\phi^{-1}\Big(1-\big(1-\phi(\textbf{E}[L(y_1)])\big)\big(1-\phi(\textbf{E}[L(y_2)])\big)\Big),\\
&\textbf{E}[L(u_2)]\\&=\textbf{E}[L(y_1)]+\textbf{E}[L(y_2)],
\end{aligned}
\end{eqnarray}
where $\textbf{E}$ denotes expectation and
\begin{eqnarray}\nonumber
& \phi(x)=\left\{
 \begin{aligned}
    &1-\frac{1}{\sqrt{4 \pi |x|}}\int_{-\infty}^{\infty}\mathrm{tanh}\frac{u}{2}e^{-\frac{(u-x)^2}{4|x|}}du, \quad x\not= 0,\\
    &1, \quad  \quad \quad\quad\quad\quad\quad\quad\quad\quad\quad\quad\quad\quad\quad x=0.\\
 \end{aligned}
 \right.
\end{eqnarray}

As the likelihood ratios (LR) are recursively calculated by equations (75) and (76) in \cite{firstArikan}, the expectation of the LLRs, i.e., $\textbf{E}[L(u_i)]$, can be calculated in a similar manner. Then, based on the assumption that $L(u_i)$ satisfies the Gaussian distribution, the error probability of each subchannel $W_N^{(i)}(y_1^N,u_1^{i-1}|u_i)$ can be calculated by using the $Q$-function as
\begin{align}
P_e(u_i)=Q(\sqrt{\textbf{E}[L(u_i)]/2}),
\end{align}
where $Q(x)=\frac{1}{\sqrt{2\pi}}\int_x^{+\infty}e^{-\frac{t^2}{2}}dt$. Since $\textbf{E}[L(u_i)]$ depends on $\textbf{E}[L(y_i)]$ with $L(y_i) \sim \mathcal{N}(\frac{2}{\sigma_n^2},\frac{4}{\sigma_n^2})$, where $\sigma_n^2$ is the noise variance, it becomes clear that $P_e(u_i)$ is also SNR dependent. In addition, it is worth pointing out that, given $\sigma_n^2$, $P_e(u_i)$ can be calculated in an off-line manner.

Having defined both the measure of reliability and the threshold, the splitting rule is given as follows:
If either of the following two inequalities holds:
\begin{equation}
\begin{aligned}
P_{l}(u_i=0|y_1^N,\hat{u}_1^{i-1})>1-P_e(u_i),
\label{P_e(u_i0)}
\end{aligned}
\end{equation}
\begin{equation}
\begin{aligned}
P_{l}(u_i=1|y_1^N,\hat{u}_1^{i-1})>1-P_e(u_i),
\label{P_e(u_i1)}
\end{aligned}
\end{equation}
the $l$-th path does not split, otherwise, the $l$-th path splits into two paths. For instance, if Eq. (\ref{P_e(u_i0)}) holds, then we directly set $\hat{u}_i=0$ instead of splitting the $l$-th path.
According to Bayes' rule, a more convenient splitting rule can be found, as follows
\begin{eqnarray}
& \hat{u}_i=\left\{
 \begin{aligned}
    &0, \,\,\, L_l(u_i)>\mathrm{log}\frac{1-P_e(u_i)}{P_e(u_i)},\\
    &1,  \,\,\, L_l(u_i)<-\mathrm{log}\frac{1-P_e(u_i)}{P_e(u_i)},\\
    &\mbox{split}, \,\,\, \mbox{otherwise}.\\
 \end{aligned}
 \right.
 \label{L(u_i)threshold}
\end{eqnarray}
where $L_l(u_i)$ denotes the LLR of $u_i$ in the $l$-th decoding path and can be calculated in a recursive manner \cite{firstArikan}. For simplicity, we drop the subscript $l$ in the ensuing analysis.

\subsection{Key Observations}
We now investigate the implications of the newly defined splitting rule. As we mainly focus on AWGN channels, the Gaussian approximation method is adopted in the ensuing analytical derivation, i.e., all the propositions in this subsection are based on the assumption that the LLR $L(u_i)$ follows Gaussian distribution. For the purpose of clear exposition, we assume that an all-zero codeword is transmitted. Please note that, according to the following proposition, using the all-zero codeword does not cause any loss of generality of the ensuing analysis, since the distribution of $L(u_i)$ is symmetric for $u_i=0$ and $u_i=1$.

\begin{proposition}
Under the Gaussian approximation, i.e., $L(u_i) \thicksim \mathcal{N}(\textbf{E}[L(u_i)],2|\textbf{E}[L(u_i)]|)$, for any codeword $u_1^N$, if $\hat{u}_1^{i-1}=u_1^{i-1}$, then we have
\begin{center}
  $\textbf{E}[L(u_i)]=\left\{
 \begin{aligned}
    &\textbf{E}[L_0(u_i)], \ \text{if} \ u_i=0\\
    &-\textbf{E}[L_0(u_i)], \ \text{if} \ u_i=1\\
 \end{aligned}
 \right.,$
\end{center}
where $\textbf{E}[L_0(u_i)]$ denotes the mean of $L(u_i)$ for the all-zero codeword transmitted over an AWGN channel with noise variance $\sigma_n^2$.
\label{proposition_mean}
\end{proposition}
\begin{proof}
See Appendix \ref{proofproposition_mean}.
\end{proof}

We start by examining the error performance of the SCL decoder with the newly defined splitting rule. The Gaussian distributed $L(u_i)$ is illustrated in Fig. \ref{Dis}. The two vertical lines correspond to two threshold values $\pm\mathrm{log}\frac{1-P_e(u_i)}{P_e(u_i)}$, cutting the entire range of $L(u_i)$ into three separate parts, i.e., $(-\infty, -\mathrm{log}\frac{1-P_e(u_i)}{P_e(u_i)})$, $[-\mathrm{log}\frac{1-P_e(u_i)}{P_e(u_i)},$ $\mathrm{log}\frac{1-P_e(u_i)}{P_e(u_i)}]$, and $(\mathrm{log}\frac{1-P_e(u_i)}{P_e(u_i)}, \infty)$. The first interval denotes the event in which no splitting is performed and $u_i$ is incorrectly decoded, i.e., $\hat{u}_i=1$. The probability of such event occurring can be computed as
\begin{equation}
\begin{aligned}
P'_e(u_i)&=\mathrm{Pr}(L(u_i)<-\mathrm{log}\frac{1-P_e(u_i)}{P_e(u_i)})\\
&=Q\big(\frac{\textbf{E}[L_0(u_i)]+\mathrm{log}(1-P_e(u_i))-\mathrm{log}(P_e(u_i))}{\sqrt{2\textbf{E}[L_0(u_i)]}}\big)\\
&=Q\big(\sqrt{\frac{\textbf{E}[L_0(u_i)]}{2}}+\frac{\mathrm{log}(1/Q(\sqrt{\frac{\textbf{E}[L_0(u_i)]}{2}})-1)}{\sqrt{2\textbf{E}[L_0(u_i)]}}\big)\\\label{perror}
&=Q\big(Q^{-1}(P_e(u_i))+\frac{\mathrm{log}(1/P_e(u_i)-1)}{2Q^{-1}(P_e(u_i))}\big).
\end{aligned}
\end{equation}
\begin{figure}[!ht]
\centering
\includegraphics[scale=0.45]{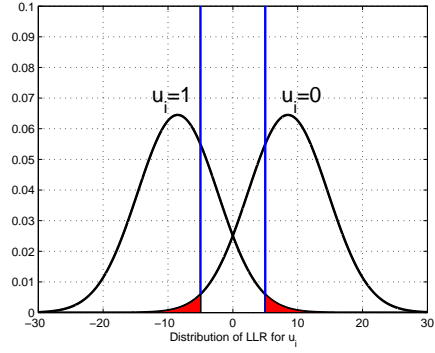}
\caption{Distribution of $L(u_i)$ under the Gaussian approximation.}
\label{Dis}
\end{figure}
Similarly, the last interval corresponds to the event in which no splitting is performed and $u_i$ is correctly decoded, i.e., $\hat{u}_i = 0$, and the probability associated with such event can be computed as
\begin{equation}\nonumber
P'_r(u_i)=\mathrm{Pr}\big(L(u_i)> \mathrm{log}\frac{1-P_e(u_i)}{P_e(u_i)}\big).\\
\end{equation}
Now, let $\mu$ and $\sigma$ be the mean and standard deviation of $L(u_i)$ when the all-zero codeword is transmitted, respectively, i.e., $\mu=\textbf{E}[L_0(u_i)]$ and $\sigma=\sqrt{2\textbf{E}[L_0(u_i)]}$. Then we have
\begin{equation}
P'_r(u_i)=Q\big(\frac{\mathrm{log}(1-P_e(u_i))-\mathrm{log}(P_e(u_i))-\mu}{\sigma}\big).\\\label{pcorrect}
\end{equation}

Since the proposed splitting rule becomes activated when the decoding reliability is high, i.e., $P_e(u_i)$ is small, it is of particular interest to see the error performance in this regime, and we have the following important results.

\begin{proposition}\label{lemmaP}
Under the Gaussian approximation, i.e., $L(u_i) \thicksim \mathcal{N}(\textbf{E}[L(u_i)],2|\textbf{E}[L(u_i)]|)$, we have 1) $\lim_{P_e(u_i)\to 0^{+}}\frac{P'_e(u_i)}{P_e(u_i)}=0$, i.e., $P'_e(u_i)=o(P_e(u_i))$, and 2) $\lim_{P_e(u_i)\to 0^{+}}P'_r(u_i)=1$.
\end{proposition}
\begin{IEEEproof}
See Appendix \ref{prooflemmaP}.
\end{IEEEproof}

\begin{figure}[!ht]
\centering
\includegraphics[scale=0.42]{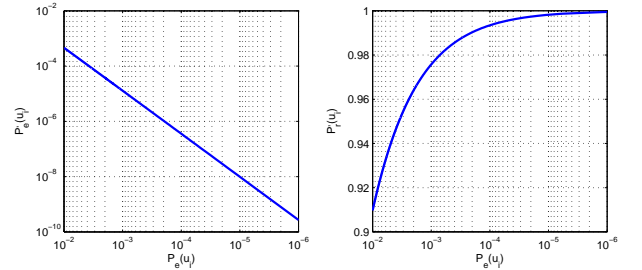}
\caption{$P'_e(u_i)$ (left) and $P'_r(u_i)$ (right) as functions of $P_e(u_i)$.}
\label{pe}
\end{figure}

The essential message of Proposition \ref{lemmaP} is that if the subchannel is sufficiently reliable, then with high probability, the correct path will not split and the unfrozen bit $u_i$ will be correctly decoded. As depicted in Fig. \ref{pe} (left), when the subchannel reliability improves, i.e., $P_e(u_i)$ becomes smaller, the decoding error $P'_e(u_i)$ decreases rapidly, and it is much smaller than $P_e(u_i)$. Similarly, Fig. \ref{pe} (right) shows that the probability of correct decoding $P'_r(u_i)$ approaches $1$ quickly when $P_e(u_i)$ becomes smaller, corroborating the claims of Proposition \ref{lemmaP}.

Armed with Proposition 2, we are ready to conjecture the behavior of the correct path in list decoding, which is in general hard to achieve a quantitative and explicit result since that number of error patterns increases exponentially and that the pruning operations involved in list decoding introduce very complicated coupling between paths.

\begin{conjecture}
Suppose that the correct decoding path survives until $u_1^{i-1}$. Under the Gaussian approximation, as $P_e(u_i)$ approaches zero, with high probability, the current path will survive at $u_{i}$ without splitting and $u_i$ will be correctly decoded. In addition, with the increasing reliability of the subsequent subchannels corresponding to $u_{i+1}^N$, the correct path will survive till termination without splitting with high probability.
\end{conjecture}

Some empirical evidences are provided in Appendix \ref{proofproposition_correct}.

Having characterized the behavior of the correct path, we now turn to examine the behavior
of the incorrect path, and we have the following conjecture:

\begin{conjecture}
Under the Gaussian approximation, for any incorrect path that survives at some unfrozen bit $u_i$, it will split at some stage within $\{i+1,i+2,...,N\}$ with high probability.
\end{conjecture}

Some empirical evidences are provided in Appendix \ref{proofproposition_incorrect}.

It was observed in \cite{Low-complexity_improved_SCL_Orion_Afisiadis} that by inverting the first erroneous bit decision, the performance of the SC decoder can be significantly improved, which implies that the decoding error occurring in $\hat{u}_1^{i-1}$ will elevate the estimate error of $u_i^{N}$ due to severe error propagation. This observation indeed provides concrete support for Conjecture 2. Now, exploiting these desirable features presented in Conjecture 1 and Conjecture 2, in combination with the proposed novel splitting rule, a low-complexity decoding procedure can be devised as detailed in the following subsection.

\subsection{The Decoding Algorithm}
The above arguments imply that all the decoding paths can be classified into two different types according to their splitting behaviors: \emph{Type a)} surviving with almost no splitting, and \emph{Type b)} splitting frequently. Ideally, a \emph{Type a)} path is unique, which corresponds to the correct codeword, while all other paths are supposed to belong to \emph{Type b)}. To reduce the decoding complexity, the key thing is to reduce the number of surviving paths at each stage. We first introduce a counter $\omega_l[i]$ for the $l$-th path at stage $i$ (corresponding to $u_i$), which counts the number of stages that the $l$-th path survives without splitting. For the $l$-th path, if it proceeds to $u_i$ without splitting, then $\omega_l[i]=\omega_l[i-1]+1$. While if the $l$-th path splits into two paths $l'$ and $l''$, then $\omega_{l'}[i]=\omega_{l''}[i]=0$. Now, utilizing the fact that the correct path seldom splits, while the incorrect path tends to split at a certain stage, we argue that if $\omega_l[i]$ exceeds a predefined threshold $\omega$, then the $l$-th path is more likely to be the correct path. As such, other remaining paths with corresponding counter values less than $\omega$ can be pruned, thereby reducing the number of surviving paths.

Under the above rationale, we propose the following split-reduced SCL Algorithm:

\begin{algorithm}
\caption{: Split-Reduced SCL Decoder}
\label{alg:A}
\begin{algorithmic}
\STATE{
\begin{enumerate}
\item[Step 1] The initialization is done by starting from the first bit $u_1$;
\item[Step 2] For the $l$-th path and unfrozen bit $u_i$, if (\ref{L(u_i)threshold}) holds, then set $\hat{u}_i$ to be $0$ or $1$ without splitting the decoding path; otherwise, split the decoding path into two paths. $\omega_l[i]$ is updated for each path in the meantime;
\item[Step 3] When the number of paths exceeds the specified list size $L$, prune those paths whose counter is less than the predetermined constant $\omega$; if no path has counter larger than $\omega$, then select the best $L$ paths according to (\ref{metric});
\item[Step 4] If $i<N$, then increase $i$ to $i=i+1$ and go to Step 2; otherwise, the candidate codeword with the smallest distance from $y_1^N$ is selected as the decoding output.
\end{enumerate}
}
\end{algorithmic}
\end{algorithm}

\subsection{Complexity and Performance Analysis}
In terms of complexity, the split-reduced SCL decoder outperforms the SCL decoder in two aspects.
\begin{enumerate}
\item Recall that for the conventional SCL decoder, the number of decoding paths doubles after each unfrozen bit is processed. Thus, the number of paths grows to the specified list size $L$ after $\mathrm{log}_2L$ unfrozen bits are processed. After which, at each stage, $2L$ paths will be pruned to obtain the surviving $L$ paths. For the split-reduced SCL decoder, as the splitting is avoided when the reliability of the subchannel is high enough, the speed of reaching the specified list size $L$ is relatively slower. In addition, $(1+\theta)L$ paths ($0 \le \theta \le 1$) are pruned on average at each stage.
\item For the conventional SCL decoder, the number of surviving paths remains fixed at $L$ after the initial $\mathrm{log}_2L$ unfrozen bits are processed. For the split-reduced SCL decoder, if the counter value of some path $l$ exceeds the predefined threshold $\omega$, the number of surviving paths can be smaller than $L$. In the extreme case, only the correct path survives while all other paths are pruned.
\end{enumerate}

It is worth pointing out that, the choice of the threshold $\omega$ affects both the decoding complexity and the error performance. In particular, it specifies the number of stages allowed at the decoder to identify the correct path. Due to the error introduced by the underlying channel, the correct path might need several stages to accumulate its reliability. During this period, there may exist incorrect paths which appears to be more reliable and are mistaken as the `correct' path by the decoder. Therefore, if $\omega$ is small, it is more likely that there will be incorrect paths exceeding the LLR threshold and does not split while the correct one keeps splitting, which leads to an irreversible loss in performance if the real correct path is eliminated. As $\omega$ increases, the correct path is given more time to accumulate its reliability, hence has a higher chance to win over the incorrect paths, thereby achieving a better performance.

Regarding the complexity, note that before any path arrives at the $\omega$ threshold, all candidate paths keep splitting with a high probability since at most one of them is the correct one. Therefore, the decoder has to wait for at least $\omega$ stages until some path achieves the threshold $\omega$. Within this period, the number of paths remains similar to that of SCL decoder. In this regard, it is desirable to have a smaller $\omega$ in terms of complexity savings.
\section{An enhanced split-reduced SCL decoder}
This section presents an enhanced split-reduced SCL decoder, by exploiting Conjecture 1, that the correct decoding path tends to survive till termination without splitting after some unfrozen bit $u_{i+\omega}$, which suggests the key idea of replacing the SCL decoder with the SC decoder after this particular unfrozen bit. Nevertheless, it is in general quite difficult to determine the exact index $i+\omega$ because the distribution of frozen and unfrozen bits is highly dependent on the underlying channel and no simple rules can be derived explicitly. However, it turns out that an upper bound, denoted by $N-K_1+1$, can be found for the index $i+\omega$, after which, splitting is completely avoided for the subsequent unfrozen bits.

\subsection{An illustrative example to determine $K_1$}
We now provide a simple polar code with block-length $N=8$ and rate $R=0.5$ as an example to illustrate how to find $N-K_1+1$. We start by constructing a full binary tree with $N=8$ leaf nodes (as mentioned in \cite{Simplified_SCL_Alamdar_Yazdi}), as shown in Fig. \ref{concrete_ideaGE}.
\begin{figure}[h]
\centering
\includegraphics[scale=0.8]{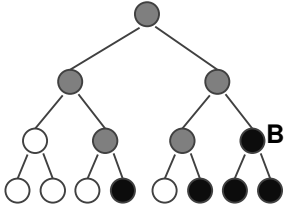}
\caption{A simple $(8,4)$ polar code with $K_1=2$.}
\label{concrete_ideaGE}
\end{figure}
Each leaf node corresponds to either a frozen bit or an unfrozen bit with an index in $\{1,2,...,N\}$ counted from left to right, and a frozen bit is denoted by a white disk while the other leaf nodes are denoted by black ones. In this particular example, $\{u_1,u_2,u_3,u_5\}$ are frozen bits while $\{u_4,u_6,u_7,u_8\}$ are unfrozen bits. Then, for a non-leaf node, if its two descendants have the same color it will also be colored the same, otherwise it is colored gray. The coloring process starts from the bottom leaf nodes until the root node is reached. After that, we start from the root node and check its right child node until the first black disk is found. In Fig. \ref{concrete_ideaGE}, we will find node $B$ which has $u_7$ and $u_8$ as its child nodes, and $K_1$ is equal to the number of leaf nodes that node $B$ has, i.e., $K_1=2$. Since $K_1$ always has an exponential form of $K_1=2^{k_1}$, instead of generating Fig. \ref{concrete_ideaGE}, one could also count the number of unfrozen bits from the last bit $u_N$ to $u_1$ until the first frozen bit is reached, the largest number of consecutive unfrozen bits, $2^{k_1}$, will be the desired $K_1$.

\subsection{SC decoding performance after $u_{N-K_1+1}$}
We now present the following important relationship between SC decoding and ML decoding after the unfrozen bit $u_{N-K_1+1}$, which will be used to design the enhanced split-reduced decoding
algorithm.

\newtheorem{theorem}{\textbf{Theorem}}
\begin{theorem}\label{theorem_MLvsSC}
Suppose that the desired $K_1$ has been found, and all the unfrozen bits with indices $\{i: 1 \le i \le N-K_1\}$ have been supplied correctly by a genie, then SC decoder will achieve exactly the same performance as ML decoder.
\end{theorem}
\begin{IEEEproof}
See Appendix \ref{prooftheorem_MLvsSC}.
\end{IEEEproof}

According to \cite{Scaling_of_polar_codes1_Hassani} and \cite{Scaling_of_polar_codes2_Hassani}, the quality of a subchannel $W_N^{(j)}$ depends heavily on the first few least significant bits of the binary expansion of $j-1$. Now, recalling the process of locating node $B$ in Fig. \ref{concrete_ideaGE}, it is observed that such a node $B$ always corresponds to a subchannel with Bhattacharyya parameter $Z_B=\left(Z(W)\right)^{2^d}$, where $d$ is the depth of node $B$. In general, $Z_B$ should take a rather small value, since $Z(W) \le 1$ and the power exponent $2^d$ grows exponentially, which implies the feasibility of using SC decoding for the unfrozen bits after $u_{N-K_1+1}$ without splitting.

\subsection{The enhanced decoding algorithm}
Based on the above observation, the enhanced split-reduced SCL decoding algorithm can then be summarized as follows:
\begin{algorithm}
\caption{:The Enhanced Split-Reduced SCL Decoder}
\label{alg:A}
\begin{algorithmic}
\STATE{
\begin{enumerate}
\item[Step 1] The initialization is done by starting from the first bit $u_1$;

\item[Step 2] For the $l$-th path and unfrozen bit $u_i$, if (\ref{L(u_i)threshold}) holds, then set $\hat{u}_i$ to be $0$ or $1$ without splitting the decoding path; otherwise, split the decoding path into two paths. $\omega_l[i]$ is updated for each path in the meantime;

\item[Step 3] When the number of paths exceeds the specified list size $L$, prune those paths whose counter is less than the predetermined constant $\omega$; if no path has counter larger than $\omega$, then select the best $L$ paths according to (\ref{metric});

\item[Step 4] If $i<N-K_1$, then increase $i$ to $i=i+1$ and go to Step 2; otherwise, simplified SC decoding is applied instead to obtain a unique estimate $(\hat{u}_{N-K_1+1},...,\hat{u}_N)$ for each surviving path, and thus the candidate codeword with the smallest distance from $y_1^N$ is selected as the decoding output.
\end{enumerate}
}
\end{algorithmic}
\end{algorithm}

Fig. \ref{Intuitive} illustrates the decoding procedure of the enhanced split-reduced SCL decoder. For the unfrozen bits before $u_{N-K_1+1}$, the splitting rule as per (\ref{L(u_i)threshold}) is used, while for the unfrozen bits $(u_{N-K_1+1},...,u_N)$, simplified SC decoding is implemented instead.

It is easy to see that the complexity is further reduced by the enhanced split-reduced SCL decoder, due to the elimination of path-splitting after $u_{N-K_1+1}$, nevertheless, the achievable error performance is not clear. In the following, we show that the enhanced split-reduced SCL decoder outperforms the original version in terms of error rate as well.

\begin{theorem}\label{theorem_enhanced_performance}
The decoding error performance achieved by the enhanced split-reduced SCL decoder is no worse than the original version.
\end{theorem}
\begin{IEEEproof}
See Appendix \ref{prooftheorem_enhanced_performance}.
\end{IEEEproof}


\begin{figure}[!ht]
\centering
\includegraphics[scale=0.48]{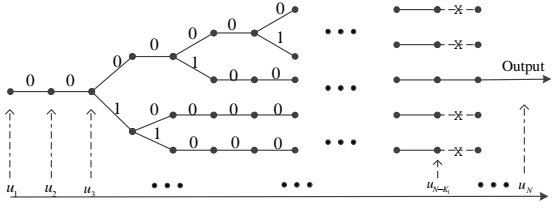}
\caption{Decoding procedure of enhanced split-reduced SCL decoding.}
\label{Intuitive}
\end{figure}

It is also of interest to consider the worst case complexity of the proposed scheme. Since the worst case appears when every path splits, where the LLR threshold and $\omega$ threshold are never achieved, hence the proposed scheme reduces to the original SCL decoder. However, the upper bound $N-K_1+1$ after which SC decoding can be implemented always exists. Therefore, even for the worst case, the proposed algorithm can reduce the complexity by $O(LK_1\mathrm{log}K_1)$ compared with the SCL decoding scheme without degrading performance.

\section{Simulation Results}
In this section, numerical simulation results are presented to illustrate the performance of the proposed decoding algorithms. Since the enhanced split-reduced SCL decoder requires lower complexity, but achieves no worse decoding error performance compared to the simple split-reduced SCL decoder, we consider only the enhanced split-reduced SCL decoder in simulations (we will use ESR-SCL as the shorthand for enhanced split-reduced SCL decoder in the following figures).

\begin{figure}[!ht]
\centering
\includegraphics[scale=0.5]{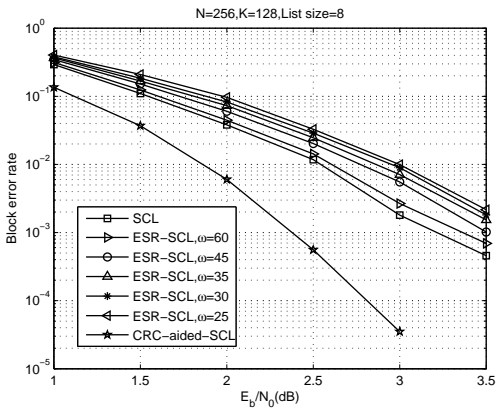}
\caption{Performance comparison of SCL decoder, enhanced split-reduced SCL decoder and CRC-aided SCL decoder.}
\label{per}
\end{figure}

Fig. \ref{per} shows the block error rate of SCL decoder, enhanced split-reduced SCL decoder with different $\omega$ and CRC-aided SCL decoder with generator polynomial $g(D)=D^{24}+D^{23}+D^6+D^5+D+1$ [18], where $N=2^8$, $K=N/2$ and list size $L=8$. As expected, when $\omega$ increases, the performance of enhanced split-reduced SCL decoder improves. In addition, we see that the CRC-aided SCL decoder significantly outperforms the SCL decoder.

\begin{figure}[!ht]
\centering
\includegraphics[scale=0.5]{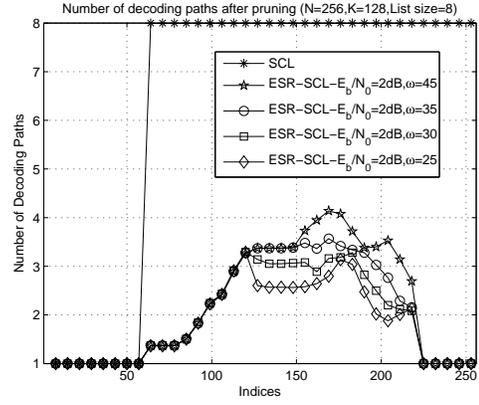}
\caption{Average number of decoding paths after pruning for $u_i$, with $E_b/N_0=2\mathrm{dB}$.}
\label{db2_omega}
\end{figure}

Define $l_i$ as the average number of decoding paths that are split when $u_i$ is processed, and let $T$ be the total number of independent trials; then $l_i\triangleq \frac{\sum_{j=1}^{T}l_{i,j}}{T}$, where $l_{i,j}$ is defined as the number of splitting paths at stage $i$ in the $j$-th experiment. The average number of paths before and after pruning are defined similarly.

\begin{figure}[!ht]
\centering
\includegraphics[scale=0.5]{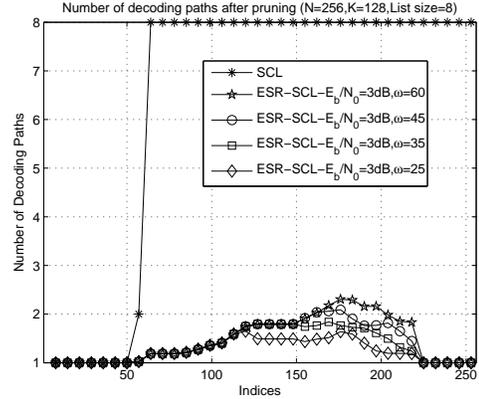}
\caption{Average number of decoding paths after pruning for $u_i$, with $E_b/N_0=3\mathrm{dB}$.}
\label{db3_omega}
\end{figure}

Fig. \ref{db2_omega} illustrates $l_i$ for SCL decoder and enhanced split-reduced SCL decoder with different $\omega$ at SNR$=2\mathrm{dB}$ when $N=2^8$, $K=N/2$ and list size $L=8$, after pruning operation. Since SCL decoder always splits decoding paths for each unfrozen bit, the number of decoding paths increases to the specified list size $L=8$ at an exponential rate, i.e., from 0 to 1, 2, 4, 8, and remains at 8 till termination. On the other hand, for enhanced split-reduced SCL decoder, it can be observed that the average number of paths after pruning operation keeps smaller than $4$ for most indices. As $\omega$ increases, less complexity can be saved, since the decoder has to wait for some longer stages until some path achieves the $\omega$ threshold, and during this period, the number of paths still stays at a large value. Besides, recall that enhanced split-reduced SCL decoder degrades to SC decoding after the index $N-K_1+1$. For this particular case, $K_1=32$, hence $K_1/K=25\%$, which indicates that 25\% of the unfrozen bits can be decoded by SC decoding rather than list decoding.

It has been observed in \cite{list_decoding_Vardy} that SCL decoder with list size $L\ge2$ almost achieves the same performance. Thus, to achieve a better performance, the list size should be at least $L=2$. Fig. \ref{db3_omega} shows the number of paths after pruning with $E_b/N_0=3\mathrm{dB}$. For small $\omega$, the average number of paths remains smaller than $2$, which implies that it can not retain some similar performance as SCL decoder. On the other hand, with large $\omega=60$, the performance significantly improves and becomes closer to that of the SCL decoder, yet with reduced complexity.

Fig. \ref{omega45_db} plots the average number of decoding paths after pruning for different SNRs with $\omega=45$. It can be observed that as SNR increases, the average number of decoding paths after pruning decreases. Since the received symbols are more reliable for high SNRs, the LLR threshold (see (\ref{L(u_i)threshold})) will be achieved with a higher probability, and once the $\omega$ threshold is achieved, the other paths will be eliminated without splitting, as analyzed by using Gaussian approximation. Besides, as the average number of paths after pruning decreases for higher SNRs, it will lead to some performance further deviating from SCL decoding (see Fig. \ref{per}).
\begin{figure}[!ht]
\centering
\includegraphics[scale=0.5]{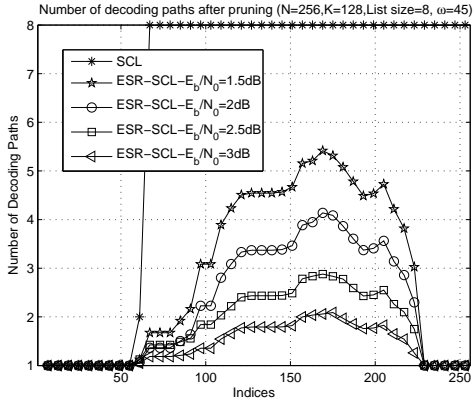}
\caption{Average number of decoding paths after pruning for $u_i$ with different SNRs.}
\label{omega45_db}
\end{figure}

\section{Conclusion}
In this paper, we have proposed low complexity split-reduced SCL decoders for polar codes. By exploiting the fact that splitting can be avoided if the reliability of decoding the unfrozen bit is high enough, a new splitting rule was defined. Under this splitting rule, it was conjectured that, if the correct path survived at some stage, it tends to survive till termination without splitting, while the incorrect path is more likely to split in the following stages. This critical behavior was then used to design a new low complexity SCL decoder. Furthermore, it was explicitly shown that there exists a particular unfrozen bit $u_{N-K_1+1}$ for any polar codes, and SC decoding can be implemented instead to decode the following unfrozen bits without degradation of error performance.

\appendix

\subsection{Proof of Proposition \ref{proposition_mean}} \label{proofproposition_mean}
We first focus on the basic decoding element defined by $G_2$, and consider the case $(u_1=1,u_2=0)$. This leads to $\textbf{E}[L(y_1)]=-\frac{2}{\sigma_n^2}$ and $\textbf{E}[L(y_2)]=\frac{2}{\sigma_n^2}$. Then,
\begin{equation}\nonumber
\begin{aligned}
\phi(\textbf{E}[L(u_1)])&=1-\big(1-\phi(-\frac{2}{\sigma_n^2})\big)\big(1-\phi(\frac{2}{\sigma_n^2})\big)\\
&\overset{(a)}=1+\big(1-\phi(\frac{2}{\sigma_n^2})\big)\big(1-\phi(\frac{2}{\sigma_n^2})\big)\\
&=2-\phi(\textbf{E}[L_0(u_1)])\overset{(b)}=\phi(-\textbf{E}[L_0(u_1)]).
\end{aligned}
\end{equation}
Steps (a) and (b) come from the fact that $\phi(x)+\phi(-x)=2$. Thus, $\textbf{E}[L(u_1)]=-\textbf{E}[L_0(u_1)]$ and $\textbf{E}[L(u_2)]=-(-\frac{2}{\sigma_n^2})+\frac{2}{\sigma_n^2}=\textbf{E}[L_0(u_2)]$. For all other possible values of $(u_1,u_2)$, similar results hold. As polar codes are recursively constructed based on $G_2$, by simple induction, the claim follows.

\subsection{Proof of Proposition \ref{lemmaP}}\label{prooflemmaP}
Denote $t=Q^{-1}(P_e(u_i))$, and thus from (\ref{perror}) we have $P'_e(u_i)=Q(t+\frac{\mathrm{log}(1/Q(t)-1)}{2t})$. For the following derivation, we will use $\overset{(L)}{=}$ to denote the L'H\^{o}pital's rule. Then, we have
\begin{equation}\nonumber
\begin{aligned}
&\lim_{P_e(u_i)\to 0^{+}}\frac{P'_e(u_i)}{P_e(u_i)}\\&=\lim_{t\to +\infty}\frac{Q(t+\frac{\mathrm{log}(1/Q(t)-1)}{2t})}{Q(t)}\\
&\overset{(L)}{=}\lim_{t\to +\infty}\frac{e^{-\frac{(t+\frac{\mathrm{log}(1/Q(t)-1)}{2t})^2}{2}}}{e^{-\frac{t^2}{2}}}\\
&=e^{-\frac{1}{2}\lim_{t\to +\infty}((\frac{\mathrm{log}(1/Q(t)-1)}{2t})^2+\mathrm{log}(1/Q(t)-1))}.
\end{aligned}
\end{equation}

One can also check that
\begin{equation}\nonumber
\begin{aligned}
\lim_{t\to +\infty}\frac{\mathrm{log}(1/Q(t)-1)}{2t}&\overset{(L)}{=}\lim_{t\to +\infty}\frac{1}{2}\frac{1}{1-Q(t)}\frac{1}{\sqrt{2\pi}}\frac{e^{-\frac{t^2}{2}}}{Q(t)}\\
&\overset{(L)}{=}\lim_{t\to +\infty}\frac{1}{2}\frac{t}{1-Q(t)}=+\infty,
\end{aligned}
\end{equation}
and
\begin{equation}\nonumber
\lim_{t\to +\infty}\mathrm{log}(1/Q(t)-1)=+\infty.
\end{equation}
Thus $\lim_{P_e(u_i)\to 0^{+}}\frac{P'_e(u_i)}{P_e(u_i)}=0$ holds.

Next, recall $P_e(u_i)=Q(\sqrt{\textbf{E}[L_0(u_i)]/2})$, i.e., $t=\sqrt{\textbf{E}[L_0(u_i)]/2}=\sigma/2$. Then we have
\begin{equation}\nonumber
\begin{aligned}
&\lim_{P_e(u_i)\to 0^{+}}P'_r(u_i)\\&=\lim_{P_e(u_i)\to 0^{+}}Q\big(\frac{\mathrm{log}(1-P_e(u_i))-\mathrm{log}P_e(u_i)-\mu}{\sigma}\big)\\
&=Q\big(\lim_{t\to +\infty}t\cdot (\frac{\mathrm{log}(1/Q(t)-1)}{2t^2}-1)\big).
\end{aligned}
\end{equation}
Note that
\begin{equation}\nonumber
\begin{aligned}
\lim_{t\to +\infty}\frac{\mathrm{log}(1/Q(t)-1)}{2t^2}&\overset{(L)}{=}\lim_{t\to +\infty}\frac{1}{4}\frac{1}{1-Q(t)}\frac{1}{\sqrt{2\pi}}e^{-\frac{t^2}{2}}\frac{1/t}{Q(t)}\\
&\overset{(L)}{=}\lim_{t\to +\infty}\frac{1}{4}\frac{1}{1-Q(t)}\frac{1}{t^2}=0,
\end{aligned}
\end{equation}
thus $\lim_{t\to +\infty}t(\frac{\mathrm{log}(1/Q(t)-1)}{2t^2}-1)=-\infty$, and we have $\lim_{P_e(u_i)\to 0^{+}}P'_r(u_i)=1$.

\subsection{Empirical evidence of Conjecture 1}\label{proofproposition_correct}
The first two statements are straightforward results due to Proposition \ref{lemmaP}. According to \cite{SCL_decoding_Chenkai}, the metric for each path could be computed in a recursive manner according to
\begin{equation}\label{metric}
\begin{aligned}
P(\hat{u}_1^{i}|y_1^N)=P(\hat{u}_1^{i-1}|y_1^N)\frac{e^{(1-\hat{u}_i)L(u_i)}}{e^{L(u_i)}+1}.
\end{aligned}
\end{equation}
For the correct path,
\begin{equation}\nonumber
\begin{aligned}
P(\hat{u}_1^{i}|y_1^N)&=P(\hat{u}_1^{i-1}|y_1^N)\frac{e^{L(u_i)}}{e^{L(u_i)}+1}\\
&\approx P(\hat{u}_1^{i-1}|y_1^N)(1-P_e(u_i))\\
&\approx P(\hat{u}_1^{i-1}|y_1^N),
\end{aligned}
\end{equation}
which implies that the reliability of this path after choosing $\hat{u}_i=0$ hardly degrades. Thus, as the correct path survives at $u_{i-1}$, it would not be pruned and continue to survive at $u_i$ with high probability. By induction on $i$, the correct path would survive to the last without splitting if the following subchannels are reliable enough.

\subsection{Empirical evidence of Conjecture 2}\label{proofproposition_incorrect}
\begin{figure}[!ht]
\centering
\includegraphics[scale=0.8]{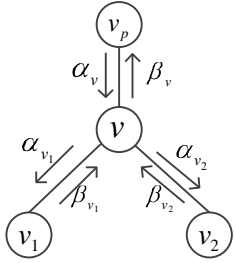}
\centering
\caption{Decoder for the constituent code.}
\label{proof}
\end{figure}
The SC decoding process can be interpreted based on a full binary tree with $N=2^n$ leaf nodes, where postorder traversal is implemented. We use Fig. \ref{proof} to give a simple illustration, where $v_1$ and $v_2$ denote the child nodes of node $v$ while $v_p$ denotes the parent node. When node $v$ is activated, it would first receive an LLR vector $\alpha_v$ from $v_p$. Suppose that the length of $\alpha_v$ is $2^p$. Then, node $v$ would compute the LLR vector $\alpha_{v_1}$ of length $2^{p-1}$ according to SC decoding and passes $\alpha_{v_1}$ to node $v_1$. After node $v_1$ produces its own codeword $\beta_{v_1}$ of length $2^{p-1}$ and passes it back to node $v$, another LLR vector $\alpha_{v_2}$ of length $2^{p-1}$ would be computed at node $v$ and sent to node $v_2$. After node $v$ receives codeword $\beta_{v_2}$, it would produce its own codeword $\beta_v$ by associating $\beta_{v_1}$ and $\beta_{v_2}$ according to $G_2$. The above description defines a recursive algorithm. The initialization is done by assigning the LLRs received from the underlying channel to the root node, while the recursion returns at each leaf node since leaf nodes correspond to the sequence $u_1^N$ and hard decisions are implemented.

Suppose that $m$ errors occur at node $v_1$, i.e., $m$ bits are set to $1$ in $\beta_{v_1}$ (assuming the all-zero codeword transmitted). With a slight abuse of notation, we use $\alpha_v[i]$ to denote the component with index $i$, and we have $\textbf{E}[\alpha_{v_2}[i]]=(1-2\beta_{v_1}[i])\textbf{E}[\alpha_{v}[2i]]+\textbf{E}[\alpha_v[2i+1]]$. Note that $\textbf{E}[\alpha_v[i]]=\textbf{E}[\alpha_v[j]]$ holds for any $1 \le i, j \le 2^p$, and thus $\textbf{E}[\alpha_{v_2}[i]]=0$ if $\beta_{v_1}[i]=1$. In (\ref{ELu}), one can check that if only $\textbf{E}[L(y_1)]$ (or $\textbf{E}[L(y_2)]$) is zero, we have $\textbf{E}[L(u_1)]=0$ and $\textbf{E}[L(u_2)]=\textbf{E}[L(y_2)]$ (or $\textbf{E}[L(u_2)]=\textbf{E}[L(y_1)]$); while if both $\textbf{E}[L(y_1)]=0$ and $\textbf{E}[L(y_2)]=0$ hold, we have $\textbf{E}[L(u_1)]=\textbf{E}[L(u_2)]=0$. Thus, the number of LLRs whose means are zero-valued remains the same after the calculation defined by (\ref{ELu}).

As node $v_2$ would pass another two LLR vectors computed according to (\ref{ELu}) to its left child node and right child node respectively, by some simple induction, we can conclude that there would be at least $m$ leaf nodes that have zero-valued means. For an unfrozen bit $u_i$, $\textbf{E}[L(u_i)]=0$ implies a significant degradation to the original subchannel, and it is more difficult to achieve the thresholds $\pm\mathrm{log}\frac{1-P_e(u_i)}{P_e(u_i)}$ ($|\mathrm{log}\frac{1-P_e(u_i)}{P_e(u_i)}|$ usually stays far away from zero if $u_i$ is an unfrozen bit). As there would be at least $m$ leaf nodes having zero-valued means, this incorrect path is quite likely to split at the following stages.

\subsection{Proof of Theorem \ref{theorem_MLvsSC}}\label{prooftheorem_MLvsSC}
We first provide a lemma, which will be invoked in the proof of Theorem \ref{theorem_MLvsSC}.
\begin{lemma}\label{lemma_ML}
For a symmetric B-DMC with received LLRs $L_1^N$, the ML decoder will output the codeword
\begin{equation}
\hat{x}_1^N=\mathop{\mathrm{argmax}}_{x_1^N \in \mathcal{C}}\sum_{i=1}^N(1-2x_i)L_i.
\label{mlequation}
\end{equation}
\end{lemma}

\begin{IEEEproof}
\begin{equation}\nonumber
\begin{aligned}
\hat{x}_1^N&=\mathop{\mathrm{argmax}}_{x_1^N \in \mathcal{C}}P(x_1^N|y_1^N)\\
&=\mathop{\mathrm{argmax}}_{x_1^N \in \mathcal{C}}\mathrm{log}P(y_1^N|x_1^N)\\
&=\mathop{\mathrm{argmax}}_{x_1^N \in \mathcal{C}}\sum_{i=1}^N\mathrm{log}P(y_i|x_i).
\end{aligned}
\end{equation}
As $y_1^N$ denotes the symbols received from the underlying channel, $\sum_{i=1}^N\mathrm{log}P(y_i|1)$ is just a constant which is independent of $x_1^N$. Thus,
\begin{equation}\nonumber
\begin{aligned}
\hat{x}_1^N &=\mathop{\mathrm{argmax}}_{x_1^N \in \mathcal{C}}\sum_{i=1}^N\mathrm{log}P(y_i|x_i)-\sum_{i=1}^N\mathrm{log}P(y_i|1)\\
&=\mathop{\mathrm{argmax}}_{x_1^N \in \mathcal{C}}\sum_{i=1}^N\mathrm{log}\frac{P(y_i|x_i)}{P(y_i|1)}\\
&=\mathop{\mathrm{argmax}}_{x_1^N \in \mathcal{C}}\sum_{i=1}^N(1-x_i)L_i\\
&=\mathop{\mathrm{argmax}}_{x_1^N \in \mathcal{C}}\big(\frac{1}{2}\sum_{i=1}^NL_i+\frac{1}{2}\sum_{i=1}^N(1-2x_i)L_i\big).
\end{aligned}
\end{equation}
Note that $\frac{1}{2}\sum_{i=1}^NL_i$ is also a constant once $y_1^N$ is determined. Thus,
\begin{equation}\nonumber
\begin{aligned}
\hat{x}_1^N=\mathop{\mathrm{argmax}}_{x_1^N \in \mathcal{C}}P(x_1^N|y_1^N)=\mathop{\mathrm{argmax}}_{x_1^N \in \mathcal{C}}\sum_{i=1}^N(1-2x_i)L_i.
\end{aligned}
\end{equation}
\end{IEEEproof}

Now we establish a full binary tree as illustrated in Fig. \ref{concrete_ideaGE} for the proof. Use $(v_1,v_2,...,v_m)$ to denote the codeword associated with a node $v$ and $(L_v[1],L_v[2],...,L_v[m])$ to denote the related LLRs. The codeword $(A_1,A_2,...,A_N)$ of root node $A$ just corresponds to $x_1^N$, the last stage output of the encoder, while $(L_A[1],L_A[2],...,L_A[N])$ represents the received LLRs from the underlying channel. Let $D$ and $C$ be the left and right child node of the root node $A$, respectively. Then considering the basic encoding operation with the matrix $G_2$ we have $A_{2i-1}=D_i\oplus C_i$ and $A_{2i}=C_i$, for $i=1,2,...,N/2$. As $u_1^{N-K_1}$ are correctly known, the ML decoder will select the estimate $(\hat{u}_{N-K_1+1},...,\hat{u}_N)$ to maximize (see Lemma \ref{lemma_ML}):
\begin{equation}\nonumber
\begin{aligned}
&\hat{u}_{N-K_1+1}^N=\mathop{\mathrm{argmax}}_{u_{N-K_1+1}^N}\sum_{i=1}^N(1-2A_i)L_A[i] \\
&=\mathop{\mathrm{argmax}}_{u_{N-K_1+1}^N}\big(\sum_{i=1}^{N/2}(1-2(D_i \oplus C_i))L_A[2i-1]\\
&+\sum_{i=1}^{N/2}(1-2C_i)L_A[2i]\big)\\
&=\mathop{\mathrm{argmax}}_{u_{N-K_1+1}^N}\sum_{i=1}^{N/2}(1-2C_i)\big((1-2D_i)L_A[2i-1]+L_A[2i]\big).
\end{aligned}
\end{equation}
The SC decoder calculates the LLRs at node $C$ according to equation (76) in \cite{firstArikan}, and one can check that $L_C[i]=(1-2D_i)L_A[2i-1]+L_A[2i]$, which is known since $L_A[i]$ is the received LLR and $D_i$ only depends on $u_1^{N-K_1}$. Thus we have
\begin{equation}\nonumber
\hat{u}_{N-K_1+1}^N =\mathop{\mathrm{argmax}}_{u_{N-K_1+1}^N}\sum_{i=1}^{N/2}(1-2C_i)L_C[i].
\end{equation}

Then we consider the child nodes of node $C$ and repeat the above steps, until the first black node $B$ is reached. Similarly, we can have
\begin{equation}
\begin{aligned}
\hat{u}_{N-K_1+1}^N =\mathop{\mathrm{argmax}}_{u_{N-K_1+1}^N}\sum_{i=1}^{K_1}(1-2B_i)L_B[i],
\label{LLR_metric_for ML_decoding}
\end{aligned}
\end{equation}
where $L_B[i]$ equals the LLR calculated by the SC decoder according to equation (76) in \cite{firstArikan}. Obviously, to maximize the summation in (\ref{LLR_metric_for ML_decoding}), it requires that the binary codeword of $(B_1,B_2,...,B_{K_1})$ are decided according to the signs of $(L_B[1], L_B[2], ..., L_B[K_1])$, which are just equivalent to the one-by-one hard decisions in SC decoding, except that an inverse encoding operation is needed to obtain the desired $\hat{u}_{N-K_1+1}^N$. Therefore, SC decoder achieves exactly the same performance as ML decoder provided the real values of $u_1^{N-K_1}$ are known.

\subsection{Proof of Theorem \ref{theorem_enhanced_performance}}\label{prooftheorem_enhanced_performance}
It is obvious that before $u_{N-K_1+1}$ is processed, the enhanced split-reduced SCL decoder achieves exactly the same performance as the original one. Suppose that $l$ paths survive when $u_{N-K_1+1}$ is reached. For each surviving path, there should be $2^{K_1}$ possible paths which all originate from the nodes at the ${(N-K_1)}$-th level (just corresponds to the ${(N-K_1)}$-th bit, see \cite{list_decoding_Vardy}) in the list decoding framework. According to Theorem \ref{theorem_MLvsSC}, for any particular path, the conventional SC decoding suffices to achieve the ML decoding performance (note that this is not the overall ML decoding performance since the estimated unfrozen bits before $u_{N-K_1+1}$ are not guaranteed to be correct). Thus, for each particular path arriving at $u_{N-K_1+1}$, the conventional SC decoding algorithm would select the best path among all $2^{K_1}$ possible ones, i.e., the best estimate $(\hat{u}_{N-K_1+1},...,\hat{u}_N)$ for each surviving path can be obtained directly. Thus, the overall best estimate of $u_1^N$ must be involved in these $l$ surviving candidate codewords. Finally, the candidate codeword that has the smallest distance from the received symbols $y_1^N$ is selected as the decoding output.

\section*{Acknowledgement}
The authors sincerely thank the Guest Editor and the anonymous reviewers for their constructive suggestions which helped us to improve the manuscript.

\bibliographystyle{IEEEtran}

\end{document}